\documentclass[prl,preprint,showpacs,preprintnumbers,amsmath,amssymb,floatfix]{revtex4}
\usepackage{amsmath, amssymb, amsthm}
\usepackage{graphics}
\usepackage{natbib}
\usepackage{bm}
\usepackage{mathrsfs}
\usepackage{verbatim}
\usepackage{wasysym}

\def\hf{{\frac{1}{2}}}

\newtheorem{lem}{Lemma}
\newtheorem{thm}{Theorem}
\newtheorem{cor}{Corollary}

\begin{document}

\title{An exact integration of a $\phi^4$ quantum field theory}
\author{Timothy D. Andersen}
\noaffiliation
\date{Received: date / Accepted: date}
\begin{abstract}
Most quantum field theories are not exactly solvable. In this paper show the statistical equivalence of the standard exponential path integral to products of Heaviside functions, i.e. a product of specially tuned uniform distributions. This allows exact integrations of certain quantum field theories. I apply the equivalence to calculate the exact, non-perturbative path integral for a 3+1-D scalar (real) phi-4 field theory.
\end{abstract}

\maketitle

The standard or ``exponential'' path-integral formulation of quantum field theory is the functional integral over a Hilbert space,
\[
\mathcal{Z} = \int D\phi e^{\frac{i}{\hbar} S[\phi]},
\] where 
$
S[\phi] = \int d^4x \mathcal{L},
$ the action functional in the quantum field $\phi$, is only exactly computable for a small number of Lagrangians such as the free field,
\[
\mathcal{L}_{free}[\phi] = \partial^\mu\phi\partial_\mu \phi - m^2\phi^2 + J\phi,
\] for a scalar field $\phi$ and source $J$ both in $\mathscr{L}_2$ Hilbert space over Minkowski spacetime \cite{Zee:2003}.

The massive, real scalar $\phi^4$ theory is one of the simplest theories in quantum field theory without an exact integration. Its Lagrangian is given by,
\[
\mathcal{L}[\phi] = \partial^\mu\phi\partial_\mu \phi - m^2\phi^2 - \lambda \phi^4 + J\phi,
\] where $\lambda>0$. There are no known mathematical techniques for calculating $\mathcal{Z}$ for the $\phi^4$ theory. The main analytical approaches are Feynman diagrams, which give perturbations of the free particle ensemble for small coupling constant, $\lambda$, or perturbation series for strong coupling and small kinetic term\cite{Kaya:2004}. The former perturbation series, however, do not converge for the quartic interaction theory and are asymptotic at best, while the latter have limited applicability in high energy physics. What we need for strong couplings and high energy is an exact solution.

Under a Wick rotation, the $3+1$-D $\phi^4$ quantum field theory becomes a $4$-D statistical mechanical theory \cite{Strominger:1983},
\[
\mathcal{Z}(J) = \int D\phi e^{-\frac{1}{\hbar} S[\phi]},
\] where
$
S[\phi] = \int d^4x \mathcal{H}
$ and
$
\mathcal{H}[\phi] = \partial^\mu\phi\partial_\mu \phi + m^2\phi^2 + \lambda \phi^4 - J\phi.
$ The functional $\mathcal{H}$ is now a 4-D energy functional, the amplitude $\mathcal{Z}$ is a partition function, and Planck's reduced constant $\hbar$ is equivalent to temperature. The path integral is equivalent to the ``canonical'' ensemble.

For experiments, we are only interested in how the quantization changes relative to the vacuum. Therefore, the quantity of interest is the ratio:
$
\mathcal{A} = \frac{Z(J)}{Z(0)}.
$
Most expectations of observables can be found from this ratio or a scaling of it.

Perturbation theory, specifically Feynman diagrams, is a time-tested and useful technique for evaluating this ratio at high energy, but it is also tedious and tends to blow up if the coupling is too strong. Therefore, there is a strong motivation in quantum field theory to calculate the ratio exactly for as many theories as possible, even ones as simple as the $\phi^4$, because of the possibility that it will lead to the development of new forms of non-perturbative calculus for making quantum predictions in QED and QCD, useful at high energies and densities where couplings become too strong. With this motivation in mind, I apply the principle of ensemble equivalence to derive an exact solution. 

The principle of ensemble equivalence is a common tool for simplifying computations in both statistical mechanics and quantum field theory. Let
\[
\Omega_N(J) = \int d\phi_1\cdots d\phi_N \delta(A_N - \sigma_N[{\bm \phi},{\bm J}]),
\] be the $N$-dimensional microcanonical ensemble and
\[
Z_N(J) = \int d\phi_1\cdots d\phi_N e^{-\sigma_N[\phi]/\hbar},
\] the $N$-dimensional canonical ensemble, with $\sigma_N$ the $N$ dimensional action. By definition, two statistical ensembles are equivalent if, in the ``thermodynamic'' limit of infinite degrees of freedom,  $N\rightarrow\infty$, they generate identical expectations of observables. Equivalence is typically established by an asymptotic relation (e.g., with the method of steepest-descent or saddle point method \cite{Horwitz:1983}) that becomes exact in the limit. The equivalence between microcanonical and canonical ensembles has been known since at least the 1930's \cite{Pauli:1973} and criteria for non-equivalence established in the 1970's \cite{Lynden:1968}\cite{Thirring:1970}\cite{Lynden:1977}. For efficiency, lattice gauge simulations frequently implement the microcanonical quantum field theory either in ``demon'' Monte Carlo as in the work of Creutz {\em et al.} \cite{Creutz:1983}\cite{Creutz:1984} or the Hamiltonian flow method of Callaway {\em et al.} \cite{Callaway:1982}\cite{Callaway:1983}. Perturbation theories for microcanonical quantum field theory have also been established \cite{Strominger:1983}\cite{Iwazaki:1985}. 

What is relatively unknown is that the microcanonical ensemble is not the only one for which an equivalence to the canonical can be shown. As I show below, a specially tuned product of uniform distributions, i.e. an ensemble density given by a product of step or Heaviside functions, can also be shown to be equivalent to the standard quantization by invoking the multidimensional method of steepest descent.

In this paper, I show the following exact equivalence:
\begin{equation}
\label{eqn:main}
\mathscr{C}\frac{Z(J)}{Z(0)} = \exp\left[\int \frac{d^4k}{(2\pi)^4} F[J]\right]
\end{equation} where 
\begin{widetext}
\begin{eqnarray}
F[J] & = & J^2\Bigg(\frac{\sqrt{3}(-h+g^{2/3})}{wg^{1/3}\sqrt{\lambda_p^{-1}(-2\alpha+h/g^{1/3}+g^{1/3})}} +\Big[\lambda_p\sqrt{-\lambda_p^{-1}(4\alpha + h/g^{1/3}+g^{1/3})}\times\nonumber\\ & &\Big(-108\frac{g}{h-2\alpha g^{1/3}+g^{2/3}}+w^{-1}\sqrt{3}(-h+g^{2/3})(h+4\alpha g^{1/3} + g^{2/3})\Big)\Big]\times\nonumber\\ & &(h+4\alpha g^{1/3} + g^{2/3})^{-2}\Bigg) \times 
\Bigg(8\Big[\sqrt{\lambda_p^{-1}(-2\alpha + h/g^{1/3} + g^{1/3})} + \sqrt{-\lambda_p^{-1}(4\alpha+h/g^{1/3} + g^{1/3})}\Big]\Bigg)^{-1},
\label{eqn:F}
\end{eqnarray} 
\end{widetext} given that
\[
g \equiv \alpha^3 + 36A\alpha\lambda_p + 6\sqrt{3}w,
\,
h \equiv \alpha^2 - 12A\lambda_p,
\]
\[
w \equiv \sqrt{A\lambda_p}(\alpha^2+4A\lambda_p),
\,
\alpha \equiv \hf(k^2 + m^2).
\] 
The value for $A$ is given by the implicit equation,
\begin{widetext}
\begin{eqnarray}
\frac{1}{\hbar} & = & \Big[\sqrt{3}\alpha^4 + 72\sqrt{3}A\alpha^2\lambda_p + 144\sqrt{3}A^2\lambda_p^2 +24\alpha^3\sqrt{A\lambda_p}+288\alpha (A\lambda_p)^{3/2} - (\sqrt{3}h + 12\alpha\sqrt{A\lambda})g^{2/3}\Big]\times\nonumber\\ & &\Big[\sqrt{\lambda_p^{-1}(-2\alpha + h/g^{1/3} + g^{1/3})} - \sqrt{-\lambda_p^{-1}(4\alpha + h/g^{1/3} + g^{1/3})}\Big]\times\nonumber\\ & &\Bigg[2\sqrt{A\lambda}g^{4/3}\sqrt{\lambda_p^{-1}(-2\alpha + h/g^{1/3} + g^{1/3})} \sqrt{-\lambda_p^{-1}(4\alpha + h/g^{1/3} + g^{1/3})}\times\nonumber\\
& &\Big(\sqrt{\lambda_p^{-1}(-2\alpha + h/g^{1/3} + g^{1/3})} + \sqrt{-\lambda_p^{-1}(4\alpha + h/g^{1/3} + g^{1/3})}\Big)\Bigg]^{-1},
\label{eqn:A}
\end{eqnarray}
\end{widetext} and $\lambda_p$ is the renormalized or ``physical'' coupling constant. The proportionality constant $\mathscr{C}$ is a trivial scaling. 

Let $S[{\bm \phi},{\bm J}] = \lim_{N\rightarrow\infty}\sigma_N$ where $\sigma_N = \epsilon \sum_{j=1}^N L_j[\phi_j,J_j]$ and, working in momentum space, $\epsilon = \left(\frac{\Delta k}{2\pi}\right)^4\sim 1/N$ is the hypercubic volume element of the 4-D momentum space which, along with $N$, is a regularization.

\begin{lem}
\label{lem:1}
Given the limit exists, the canonical ensemble and an ensemble defined by the product of Heaviside functions are equivalent for particular $A_j$, i.e.,
\[
\mathscr{C}\frac{Z[J]}{Z[0]} = \lim_{N\rightarrow\infty}\prod_{j=1}^N\frac{\int d\phi_j\, \theta(A_j - \epsilon L_j[\phi_j,J_j])}{\int d\phi_j\, \theta(A_j - \epsilon L_j[\phi_j,0])},
\] where $\mathscr{C}$ is a trivial proportionality constant. The constant is trivial if and only if, given a Hilbert space functional $O=\lim_{N\rightarrow\infty} O_N$ such that $O_N\sim O(1)$,
\begin{eqnarray}
\langle O \rangle & = & \lim_{N\rightarrow\infty}\frac{\int d\phi_1\cdots d\phi_N O_N\rho_\theta(\phi_1,\dots,\phi_N)}{\int d\phi_1\cdots d\phi_N \rho_\theta(\phi_1,\dots,\phi_N)}\nonumber\\ & = & \lim_{N\rightarrow\infty}\frac{\int d\phi_1\cdots d\phi_N O_N\rho_{\exp}(\phi_1,\dots,\phi_N)}{\int d\phi_1\cdots d\phi_N \rho_{\exp}(\phi_1,\dots,\phi_N)},
\end{eqnarray} for $\rho_\theta=\prod_j \theta(A_j - \epsilon L_j)$ and $\rho_{\exp} = \prod_j e^{-\epsilon L_j/\hbar}$ for some choice of $A_j$.
\end{lem}
\begin{proof}
The Heaviside has numerous representations, but the best one to use here is,
\[
\theta(x-x_0) = \int_0^x dt\, \delta(t - x_0) = \int_0^x dt \int_{\gamma-i\infty}^{\gamma+i\infty} \frac{d\eta}{2\pi i} e^{-\eta(x_0-t)},
\] where we have used the inverse Laplacian representation for the Dirac delta function \cite{Horwitz:1983}. 

Now, given that
\begin{eqnarray}
& &\int d\phi_j\theta(A_j - \epsilon L_j[\phi_j,J_j]) =\nonumber\\ & & \int d\phi_j \int_0^{A_j} da_j\int_{\gamma-i\infty }^{\gamma +i\infty} \frac{d\eta_j}{2\pi i}e^{\eta_j (a_j - \epsilon L_j[\phi_j,J_j])}\nonumber,
\end{eqnarray} we can move the integrals over $\phi_j$ inside:
\begin{eqnarray}
& &\prod_j \int_0^{A_j} da_j\int_{\gamma-i\infty }^{\gamma +i\infty} \frac{d\eta_j}{2\pi i} e^{\eta_j a_j} \int d\phi_j e^{-\eta_j\epsilon L_j[\phi_j,J_j])} = \nonumber\\ & & 
\prod_j \int_0^{A_j} da_j\int_{\gamma-i\infty }^{\gamma +i\infty} \frac{d\eta_j}{2\pi i} e^{\eta_j a_j} e^{-\eta_j F_j},
\end{eqnarray} where $e^{-\eta_j F_j} = \int d\phi_j e^{-\eta_j\epsilon L_j[\phi_j,J_j]}$. Note that $F_j$ is a function of $\eta_j$.

The method of steepest descent applies here because the ``entropy'' expression increases with $N$, $S_N =\sum_j \eta_j(a_j-F_j)\sim O(N)$. Thus, the exponent increases with $N$, and the integral can be expanded asymptotically in $N$. The $N$-vector saddle-point ${\bm \eta}_0$, where $\partial/\partial {\bm \eta} S_N = 0$ is the maximal point of the entropy. As $N\rightarrow \infty$, we obtain the first order asymptotic relation,
\begin{equation}
\label{eqn:asym}
\int_{\gamma-i\infty }^{\gamma +i\infty} \frac{d\eta_1}{2\pi i}\cdots\int_{\gamma-i\infty }^{\gamma +i\infty} \frac{d\eta_N}{2\pi i} e^{S_N[{\bm \eta}]} \approx C_N e^{S_N[{\bm \eta}_0]},
\end{equation} where $C_N$ is related to the Hessian matrix of $S_N/N$,
\[
C_N = \left(\frac{2\pi}{N}\right)^{N/2}\left[\det -N^{-1}\frac{\partial^2 S_N}{\partial\eta_i\eta_j}\right]^{-1/2}.
\] Since $S_N[{\bm \eta}_0] = {\bm \eta}_0^T {\bm a} - {\bm \eta}_0^T {\bm F}$, normalizing,
\[
\frac{\int_0^{{\bm A}} da_1\cdots da_N C_N[J] e^{{\bm \eta}_0^T {\bm a} - {\bm \eta}_0^T {\bm F}[J]}}{\int_0^{{\bm A}} da_1\cdots da_N C_N[0] e^{{\bm \eta}_0^T {\bm a} - {\bm \eta}_0^T {\bm F}[0]}} = \frac{C_N[J]e^{- {\bm \eta}_0^T {\bm F}[J]}}{C_N[0]e^{- {\bm \eta}_0^T {\bm F}[0]}}.
\] 

With all ensemble equivalence proofs, trivial scaling factors vanish in normalization (see pg. 36 of \cite{Pauli:1973}). When calculating the expectations of observables via method of steepest descent in the thermodynamic limit, i.e,
\[
\langle O \rangle = \lim_{N\rightarrow\infty}\frac{\int d\phi_1\cdots d\phi_N O_N\rho(\phi_1,\dots,\phi_N)}{\int d\phi_1\cdots d\phi_N \rho(\phi_1,\dots,\phi_N)},
\] for $\rho=\prod_j \theta(A_j - \epsilon L_j)$, the constant $C_N$ cancels with itself (a consequence of $O_N\sim O(1)$), and we obtain the exponential path integral. Hence, the constant is unessential. Let $\mathscr{C} = \lim_{N\rightarrow\infty} C_N[J]/C_N[0]$.

Since we are free to choose $A_j$, we can select a positive, real value for $A_j$, called $A^0_j$, so that the saddle point of $S_N$ is at ${\bm \eta}_0= (1/\hbar,\dots,1/\hbar)$. (The value of $A^0_j$ is not necessarily the same for every degree of freedom, hence the subscript $j$.) In taking the limit, the asymptotic equation \ref{eqn:asym} becomes exact, and the theorem is proved.
\end{proof}

As a corollary, I evaluate the free particle theory,
\begin{cor}
\begin{eqnarray}
\lim_{N\rightarrow \infty} & &\frac{\int d\phi_1\cdots d\phi_N\, \exp[- \sum_j \epsilon L_j[{\bf \phi},{\bf J}]/\hbar]}{\int d\phi_1\cdots d\phi_N\, \exp[- \sum_j \epsilon L_j[{\bf \phi},0]/\hbar]} = \nonumber\\& &\lim_{N\rightarrow\infty}\prod_{j=1}^N\frac{\int d\phi_j\, \theta(A_j - \epsilon L_j[\phi_j,J_j])}{\int d\phi_j\, \theta(A_j - \epsilon L_j[\phi_j,0])}
\end{eqnarray} where $A_N=\sum_j A_j = \hbar N/2$ and the Lagrangian density is given by $L_j[\phi_j,J_j] = \hf(k_j^2+m^2) \phi_j^2 - J_j\phi_j$.
\end{cor}
\begin{proof}
The corollary follows from direct integration, followed by taking the limit. Given $\epsilon = \left(\frac{\Delta k}{2\pi}\right)^4\sim 1/N$, one can easily show that, in the limit, the right hand side is
\begin{eqnarray}
& &\lim_{N\rightarrow\infty}\prod_{j=1}^N\frac{\int d\phi_j\, \theta(A_j - \epsilon L_j[\phi_j,J_j])}{\int d\phi_j\, \theta(A_j - \epsilon L_j[\phi_j,0])} = \nonumber \\ & &\lim_{N\rightarrow\infty}\prod_{j=1}^N \left(1 + \frac{J_j^2}{2(k^2+m^2)A_j}\left(\frac{\Delta k}{2\pi}\right)^4\right)^{1/2} = \nonumber\\ & &\exp\left[\int \left(\frac{dk}{2\pi}\right)^4 \frac{J^2}{4(k^2+m^2)A_j}\right],
\end{eqnarray} using $\sqrt{1+x} = 1 + \hf x + O(x^2)$ and the product integral. Let $A_j=\hbar/2$, and we have the accepted solution.
\end{proof}

The solution with a $\phi^4$ interaction term involves an exact solution to a (depressed) quartic equation.

\begin{thm}
\label{thm:1}
\begin{equation}
\mathscr{C}\frac{Z(J)}{Z(0)} = \exp\left[\int \frac{d^4k}{(2\pi)^4} F[J]\right]
\end{equation} where $F$ is given by \ref{eqn:F}.
\end{thm}
\begin{proof}
By Lemma \ref{lem:1}, there exists an $A_j$ such that:
\begin{equation}
\mathscr{C}\frac{Z(J)}{Z(0)} = \lim_{N\rightarrow\infty}\prod_{j=1}^N\frac{\int d\phi_j\, \theta(A_j - L[\phi_j,J_j])}{\int d\phi_j\, \theta(A_j - L_j[\phi_j,0])}.
\label{eqn:thm1}
\end{equation}

For a $\phi^4$ theory in a 4-D Euclidean space,
\[
L_j[\phi_j,J_j] = \left(\frac{\Delta k}{2\pi}\right)^4\left[\hf(k_j^2 + m^2)\phi_j^2 + \lambda\phi_j^4 - J_j\phi_j\right].
\]

Given that $k_j^2+m^2, \lambda, A_j>0$, by Descartes' rule of signs, $L_j[\phi_j,J_j]$ has exactly two real roots in $\phi_j$. Let these roots be $\phi_j^0$ and $\phi_j^1$ such that $\phi_j^1>\phi_j^0$, then $L[\phi_j,J_j] < A_j$ for $\phi_j\in (\phi_j^0,\phi_j^1)$.

The integral,
\begin{equation}
\int d\phi_j\, \theta(A_j - L_j[\phi_j,J_j]),
\end{equation} evaluates to the sum of lengths of the intervals such that $L_j[\phi_j,J_j] < A_j$. Therefore, since there is only one interval, it evaluates to the difference between the two roots:
\[
\phi_j^1 - \phi_j^0 = \int d\phi_j\, \theta(A_j - L_j[\phi_j,J_j]).
\] 

Let $\alpha=\hf(k_j^2+m^2)$, $\beta=-J_j$, and $\gamma=-\frac{A_j}{\epsilon}$ where $\epsilon = \left(\frac{\Delta k}{2\pi}\right)^4$. Let $x=\phi_j$. We need to solve:

\[
\lambda x^4 + \alpha x^2 + \beta x + \gamma = 0.
\]

This is a depressed quartic and has a well-known solution which can be obtain by, e.g., Ferrari's method. The solution is too lengthy to write out here but is easily obtained. I select the two real solutions and subtract them, then, setting $J=0$ in the subtraction, get the ratio:
\[
R_j = \frac{\phi_j^1 - \phi_j^0}{(\phi_j^1 - \phi_j^0)_{J=0}}
\]

Calculating the first order asymptotic expression of $R_j$ in $\epsilon$, we want an expression of the form
\[
R_j = 1 + \epsilon F_j + O(\epsilon^q),
\] with $q>1$ that will allow us to take the product integral. We do not obtain this form unless we renormalize the coupling constant $\lambda$ with respect to the regularization so that it becomes $\lambda=\lambda_p \epsilon$, where the subscript $p$ stands for ``physical'' in that $\lambda_p$ is the physically measured constant \cite{Zee:2003}.  With this modification, our asymptotic expression has the required form. (Because of the tediousness of the calculations, I obtain the expression using mathematical software.)

For our definition of $\epsilon$, a product integral has the definition,

\[
\lim_{N\rightarrow\infty}\prod_{j=1}^N 1 + \epsilon F_j + O(\epsilon^q) = e^{\int \frac{d^4k}{(2\pi)^4} F(k)}
\]

Thus, if we carry out these calculations,
\[
\lim_{N\rightarrow\infty} \prod_j R_j = \exp\left[\int \frac{d^4k}{(2\pi)^4} F[J]\right].
\] By Lemma \ref{lem:1}, the theorem is proved.
\end{proof}

The only remaining step is to determine $A_j^0$ in terms of $\hbar$, i.e. we must ``tune'' our uniform distributions. We can show via the steepest-descent method in the proof of Theorem \ref{thm:1},
\[
\lim_{N\rightarrow\infty}\frac{\partial S_N'}{\partial A_j} = \frac{1}{\hbar},
\] where $S_N' = \log \Omega_N$, gives the correct values. (The prime indicates this is the microcanonical entropy which is equivalent to the canonical only in the limit.) For the $\phi^4$ theory,
\[
\Omega_N = \prod_j \phi_j^1 - \phi_j^0.
\]

Thus, we must solve,
\[
\lim_{N\rightarrow\infty} \frac{\partial\log(\phi_j^1 - \phi_j^0)}{\partial A_j} = \frac{1}{\hbar},
\] for each $A_j$. In the free particle case, one finds that $A_j^0 = \hbar/2$ for all $j$. In the $\phi^4$ case, however, we arrive at an $A_j^0$ that depends on the input parameters $m$ and $\lambda_p$ as well as the wavenumber, $k_j$. Taking the limit, we arrive at equation \ref{eqn:A} (where the sub- and superscripts have been dropped, $A_j^0 \rightarrow A$, for readability), which must be solved by, e.g., Newton's method. Note that the equation is constant in $J$; hence the value of $A$ needs to be computed only once for all input sources for given $m$, $k$, and $\lambda_p$.

I have derived an exact ``closed form'' solution to the quantization of the non-perturbative scalar, real valued $\phi^4$ theory. I put ``closed form'' in quotation marks because in reality we still have to evaluate the momentum integral, but we do not have an infinite sequence of integrals as in perturbation theory. The integral is an order $N$ computation (where $N$ is the number of lattice points for momenta discretization) and can be done numerically. In the case of point $J$, i.e. pure state source and sink, the integral may be evaluated exactly. In practice, expectations of observables can be calculated directly on the basis of ensemble equivalence using the density function $\rho_\theta=\lim_{N\rightarrow\infty}\prod_j \theta(A_j - \epsilon L_j)$ or by derivatives of \ref{eqn:main}.

\bibliography{sg2}

\end{document}